\documentclass[smallextended]{svjour3}

\smartqed

\usepackage{color}
\usepackage{amsfonts,amsmath,amssymb,amsfonts}
\usepackage{graphicx}
\usepackage{tikz}
\usepackage{subfigure}
\usepackage{booktabs}
\usepackage{wasysym}
\usepackage{enumerate}
\usepackage{calc}

\usetikzlibrary{arrows,positioning,matrix}
\tikzstyle{line} = [draw, -]

\newcommand{\sign}{{\rm sign}}
\newcommand{\csign}{{\rm csign}}
\newcommand{\spec}{{\rm spec}}
\newcommand{\Dslash}{\mathcal{D}}
\newcommand{\DslashL}{\mathcal{D}_{\mathcal{L}}}

\newcommand{\pla}{Q}
\newcommand{\mathC}{\mathbb{C}}

\newcommand{\muh}{\hat{\mu}}
\newcommand{\nuh}{\hat{\nu}}
\newcommand{\SU}{\mathrm{SU}}
\newcommand{\su}{\mathfrak{su}}
\newcommand{\Sw}{S_\mathrm{W}}

\newcommand{\tr}[1]{\mathrm{tr}\left(#1\right)}
\newcommand{\dxmu}{\partial_{x,\mu}}
\renewcommand{\Re}{\mathrm{Re}}
\renewcommand{\outer}{\mathit{outer}}
\newcommand{\inner}{\mathit{inner}}

\newcommand{\restart}{\mathit{restart}}
\renewcommand{\prec}{\mathit{prec}}
\newcommand{\default}{\mathit{def}}
\newcommand{\ovl}{\mathit{ker}}

\newcounter{confcounter}
\newcommand{\conflabel}[1]{\refstepcounter{confcounter} \label{#1}}
 
\newcommand{\norm}[2][2]{\| #2 \|_{#1}}

\newcommand{\Q}[2]{\raisebox{-3ex}{\begin{tikzpicture}[scale=.4]\draw[thin,
black!30]
        (-1.2,-1.2) grid (1.2,1.2);\draw[->,shorten >=1pt]
        ({max(0,#1)},{max(0,#2)}) -- 
({max(0,#1)-1},{max(0,#2)});\draw[->,shorten >=1pt]
        ({max(0,#1)-1},{max(0,#2)}) -- 
({max(0,#1)-1},{max(0,#2)-1});\draw[->,shorten >=1pt]
        ({max(0,#1)-1},{max(0,#2)-1}) -- 
({max(0,#1)},{max(0,#2)-1});\draw[->,shorten >=1pt]
        ({max(0,#1)},{max(0,#2)-1}) -- 
({max(0,#1)},{max(0,#2)});\end{tikzpicture}}}

\begin{document}

\title{Multigrid Preconditioning for the Overlap Operator in Lattice QCD
  \thanks{This work was partially funded by the Deutsche
    Forschungsgemeinschaft (DFG) Transregional Collaborative Research
    Centre 55 (SFB/TRR55) and by the National Science Foundation under
    grant DMS:1320608}
}


\author{James Brannick \and 
        Andreas Frommer \and
        Karsten Kahl \and 
        Bj\"orn Leder \and
        Matthias Rottmann \and
        Artur Strebel
}


\institute{J. Brannick \at
              Department of Mathematics, Pennsylvania State University \\
              \email{brannick@psu.edu}           
           \and
           A. Frommer, K. Kahl, B. Leder, M. Rottmann, A. Strebel \at
              Fachbereich Mathematik und Naturwissenschaften, Bergische Universit\"at Wuppertal \\
              \email{\{frommer,kkahl,leder,rottmann,strebel\}@math.uni-wuppertal.de}
}

\date{Received: \today / Accepted: }

\maketitle

\begin{abstract}
The overlap operator is a lattice discretization of the Dirac operator 
of quantum chromodynamics, the fundamental physical theory of the strong interaction between the quarks. 
As opposed to other
discretizations it preserves the important physical property of chiral symmetry, at the expense of requiring
much more effort when solving systems with this operator. We present a preconditioning technique based on another lattice discretization, the Wilson-Dirac operator. The mathematical analysis precisely describes the effect of this preconditioning in the case that the Wilson-Dirac operator is normal. Although this is not exactly the case in realistic settings, we show that 
current smearing techniques indeed drive the Wilson-Dirac operator towards normality, thus providing a motivation why our preconditioner works well in computational practice. Results of numerical experiments 
in physically relevant settings show that our preconditioning yields accelerations of up to one order of magnitude. 

\end{abstract}


\section{Introduction}\label{sec:introduction}

The purpose of this paper is to motivate, analyze and experimentally validate a new preconditioner 
for the overlap operator of lattice QCD. Lattice QCD simulations are among today's most demanding supercomputer applications \cite{PRACE:ScAnnRep12,PRACE:ScC12} and substantial resources are spent in these computations.
From a theoretical point of view, the overlap operator is particularly attractive since it respects an
important physical property, chiral symmetry, which is violated by 
other lattice discretizations. From a practical point of view, the overlap operator has the disadvantage that its computational cost can be two orders of magnitude larger than when using standard discretizations.  

The basic idea of the preconditioner we propose is to use a standard discretization of the Dirac equation to form a preconditioner for the overlap operator. This may be regarded as a variant of the fictitious (or auxiliary) space preconditioning technique \cite{Nepomnyaschikh1991}
that has been used for developing and analyzing multilevel preconditioners for various nonconforming
finite element approximations of PDEs; cf.~\cite{Oswald1996,Xu1996}. In this context, one works with a mapping from the 
original space to a fictitious space, yielding an equivalent problem that is easier 
to solve. Preconditioning is then done by (approximately) solving this equivalent problem. The convergence
properties of auxiliary space preconditioning depend on the choice of the 
fictitious space, and its computational efficiency depends, in addition, on the efficiency of the solver used in that space; cf.\ \cite{Nepomnyaschikh1991}.

For the overlap operator in lattice QCD, choosing its
kernel---the Wilson-Dirac operator---as the auxiliary space preconditioner is facilitated by the fact that both operators 
are defined on the same Hilbert space.  In this way, the preconditioner for the former can be constructed 
using an adaptive algebraic multigrid solver for the latter on the same finite dimensional lattice.
We note that  similar approaches are possible for other QCD discretizations.
For example, the direct and strong coupling of the Wilson blocks used in the 5d domain wall operator \cite{Kaplan1992342} suggest that
a similar Wilson auxiliary-space preconditioner (with a more general mapping) may also be effective.

We demonstrate that the technique we develop in this paper is able to reduce the computational cost for solving systems with the overlap operator substantially, reaching speed-ups of a factor of 10 or more in realistic settings.  The preconditioning technique thus contributes to making the overlap operator more tractable in lattice QCD calculations.  

This paper is organized as follows. We start by explaining some physical background in section~\ref{QCD:sec} where we also introduce the two lattice discretizations of the continuum Dirac equation of interest here, the Wilson-Dirac operator and the overlap operator. In section~\ref{prec:sec} we give a precise mathematical analysis which shows that our preconditioner is effective in an idealized setting where both operators are assumed to be normal. Section~\ref{smearing:sec} shows that current {\em smearing} techniques in lattice QCD can be viewed as methods which drive the discretizations towards normality, thus motivating that we can expect the analysis of the idealized setting to also reflect the influence of the preconditioner in realistic settings. This is then confirmed by large-scale parallel numerical experiments reported in section~\ref{numerics:sec} which are performed for lattice configurations coming from state-of-the-art physical simulations.

\section{The Wilson-Dirac and the Overlap Operator in Lattice QCD}
\label{QCD:sec}
Quantum Chromodynamics (QCD) is a quantum field theory for the strong 
interaction of the quarks via gluons and as such part of the standard model of 
elementary particle physics. Predictions that can be deduced from this theory 
include the masses and resonance spectra of hadrons---composite particles 
bound by the strong interaction (e.g., nucleon, pion; cf.~\cite{Durr21112008}). 

The Dirac equation
\begin{equation}\label{Dirac_eq}
  (\Dslash+m)\psi = \eta
\end{equation} 
is at the heart of QCD. It describes the dynamics of the quarks and the 
interaction of quarks and gluons. 
Here, $\psi 
= \psi(x)$ and $\eta = \eta(x)$ represent quark fields. They depend on $x$, the 
points in space-time, $x=(x_0,x_1,x_2,x_3)$\footnote{Physical space-time is a 
four-dimensional Minkowski space. We present the theory in Euclidean space-time 
since this version can be treated numerically. The two versions are equivalent, 
cf.~\cite{montvay1994quantum}.}. The gluons are represented in the Dirac 
operator $\Dslash$ to be discussed below, and $m$ is a scalar mass parameter. 
It is independent of $x$ and sets the mass of the 
quarks in the QCD theory.

$\Dslash$ is given as 
\begin{equation} \label{Dirac_continuum:eq}
\Dslash=\sum_{\mu=0}^3\gamma_\mu \otimes \left( \partial_\mu + A_\mu \right)\,,
\end{equation}
where $ \partial_\mu = \partial / \partial x_\mu$ and $A$ is the gluon (background) gauge field with 
the anti-hermitian traceless matrices $A_\mu(x)$ being elements of $\mathfrak{su}(3)$, the Lie 
algebra of the special unitary group $\mathrm{SU}(3)$. The $\gamma$-matrices 
$\gamma_0,\gamma_1,\gamma_2,\gamma_3 \in  \mathbb{C}^{4 \times 4}$ represent the generators of the 
Clifford algebra with
\begin{equation} \label{commutativity_rel:eq}
  \gamma_\mu \gamma_\nu + \gamma_\nu \gamma_\mu = \begin{cases} 2 \cdot \mbox{id} &\mu = \nu\\0 & \mu \neq \nu \end{cases} \quad \text{ for } \mu,\nu=0,1,2,3. 
\end{equation}
Consequently, at each point $x$ in space-time, the spinor $\psi(x)$, i.e., the quark field $\psi$ at 
a given point $x$, is a twelve component column vector, each component corresponding to one of 
three colors (acted upon by $A_\mu(x)$) and four spins (acted upon by $\gamma_\mu$).  

For future use we remark that $\gamma_5 = 
\gamma_0\gamma_1\gamma_2\gamma_3$ satisfies
\begin{equation} \label{gamma_commutativity:eq}
     \gamma_5 \gamma_\mu = - \gamma_\mu \gamma_5, \enspace \mu=0,1,2,3.
\end{equation}



The only known way to obtain predictions in QCD from first principles and non-perturbatively, is to 
discretize and then simulate on a computer. The discretization is typically formulated on an equispaced lattice.
In a lattice discretization, a periodic $N_t \times N_s^3$ lattice $\mathcal{L}$ with uniform lattice spacing $a$ is used, $N_s$  denoting the number of lattice points for each of the three space dimensions and $N_t$ the number of lattice points in the time dimension. A quark field $\psi$ is now represented by its values at each lattice point, i.e., it is
a spinor valued function $\psi: \mathcal{L} \to \psi(x) \in \mathbb{C}^{12}$.

The {\em Wilson-Dirac} discretization is the most commonly used discretization in lattice QCD simulations. It is obtained from
the continuum equation by replacing the covariant derivatives by centralized covariant finite differences on the lattice together with an additional second order finite difference stabilization term. The Wilson-Dirac discretization yields a local operator in the sense that it represents a nearest neighbor coupling on the lattice. To precisely describe the action of the Wilson-Dirac operator $D_W$ on a (discrete) quark field $\psi$ we introduce the shift vectors $ \hat{\mu} = (\hat{\mu}_0,\hat{\mu}_1, \hat{\mu}_2, \hat{\mu}_3)  \in \mathbb{R}^4 $ in dimension $\mu$ on the lattice, i.e., 
$$
  \hat{\mu}_{\nu} = \begin{cases} a & \mu=\nu \\ 0 & \text{else.} \end{cases}.
$$
Then 
\begin{eqnarray}
  \hspace*{-0.7em}(D_W\psi)(x) \, = \,  \frac{m_0+4}{a} \psi(x) 
              &-& \frac{1}{2a}\sum_{\mu=0}^3 \left( (I_4-\gamma_\mu)\otimes U_\mu(x)\right) \psi(x+\hat{\mu}) \nonumber \\
             &-& \frac{1}{2a}\sum_{\mu=0}^3 \left( (I_4+\gamma_\mu)\otimes U_\mu^H(x-\hat{\mu})\right) \psi(x-\hat{\mu}), \label{Wilson-Dirac:eq}
\end{eqnarray}
where $U_\mu(x)$ now are matrices from the Lie group SU(3), and the lattice indices $x\pm \hat{\mu}$ are to be understood periodically. 
The mass parameter $m_0$ sets the quark mass (for further details, see~\cite{montvay1994quantum}), and we will write $D_W(m_0)$ whenever the dependence on $m_0$ is important. The matrices $U_\mu(x)$ are called  
{\em gauge links}, and the collection $\mathcal{U} = \{U_\mu(x): x \in \mathcal{L}, \mu=0,\ldots3\}$ is termed the {\em gauge field}.

From \eqref{Wilson-Dirac:eq} we see that the
couplings in $D_W$ from lattice site $x$ to $x+\hat{\mu}$ and from $x+\hat{\mu}$ to $x$ are given by 
\begin{equation} \label{DandDH_entries:eq}
(D_W)_{x,x+\hat{\mu}} = -(I_4-\gamma_\mu) \otimes U_\mu(x) \enspace \mbox{and} \enspace
(D_W)_{x+\hat{\mu},x} =  -(I_4+\gamma_\mu) \otimes U^H_\mu(x),
\end{equation}
respectively. Due to the commutativity relations \eqref{gamma_commutativity:eq} we therefore have that
\[
(\gamma_5\otimes I_3) \big(D_W\big)_{x,x+\hat{\mu}} = (\gamma_5 \otimes I_3) \big(D_W\big)_{x+\hat{\mu},x},
\]
implying that with $\Gamma_5 = I_{n_{\mathcal{L}}} \otimes \gamma_5 
\otimes I_3 $, $n_\mathcal{L}$ the number of lattice sites, we have 
\begin{equation} \label{gamma_5_symmetry:eq}
(\Gamma_5 D_W)^H = \Gamma_5 D_W.
\end{equation}
This $\Gamma_5$-symmetry is a non-trivial, fundamental symmetry
which the discrete Wilson-Dirac operator inherits from a corresponding
symmetry of the continuum Dirac operator
\eqref{Dirac_continuum:eq}. The matrix $\Gamma_5$ is hermitian and
unitary, since $\gamma_5^H = \gamma_5$ and $\gamma_5^2 = id$; see
\cite{FroKaKrLeRo13}, e.g., and \eqref{gamma_commutativity:eq}. 

The Wilson-Dirac operator and its clover improved variant (where a diagonal term is added in order 
to reduce the local discretization error from $\mathcal{O}(a)$ to $\mathcal{O}(a^2)$) is an adequate 
discretization for the numerical computation of many physical observables. It, however, breaks 
another fundamental symmetry of the continuum operator, namely {\em
  chiral symmetry}, which is of vital importance for some physical
observables like hadron spectra in the presence of magnetic fields,
for example.  As was pointed out in \cite{Luscher:1998pqa}, 
a lattice discretization $D$ of $\Dslash$ which obeys the Ginsparg-Wilson relation~\cite{WilsonGinsparg1982}
\begin{equation} \label{GW:eq}
\Gamma_5D + D\Gamma_5 = aD\Gamma_5 D
\end{equation}
satisfies an appropriate lattice variant of chiral symmetry. It has
long been unknown whether such a discretization exists until Neuberger
constructed it in~\cite{Neuberger1998141}. For convenience, the essentials of the arguments 
in \cite{Neuberger1998141} are summarized in the following proposition and its proof.

\begin{proposition} \label{neuberger:prop}
\emph{Neuberger's overlap operator}
  \begin{equation*}
    D_N = \frac{1}{a} \left( \rho I + D_{W}(m^{\ovl}_{0})\Big(D_{W}(m^{\ovl}_{0})^H
    (D_{W}(m^{\ovl}_{0})\Big)^{-\frac12} \right)
  \end{equation*} fulfills~\eqref{GW:eq} for $\rho=1$, has local discretization
  error $\mathcal{O}(a)$, and is a stable discretization
  provided $-2 < m_0^\ovl < 0$.
  \end{proposition}
  \begin{proof} We write $\DslashL$ for the 
  restriction of the continuum Dirac operator $\Dslash$ from \eqref{Dirac_continuum:eq} to the lattice $\mathcal{L}$, i.e.,
  $\DslashL$ is the finite dimensional operator which takes the same values as $\Dslash$ at the points from $\mathcal{L}$. 
  The fact that the Wilson-Dirac operator has first order discretization error can then be expressed as\footnote{For simplicity, we consider here the ``naive'' limit $a \rightarrow 0$. In the full quantum theory
  one has $\DslashL = D_W(m_0(a)) + \mathcal{O}(a)$ with the mass $m_0(a)$ of order $1/\log(a)$; see~\cite{montvay1994quantum}.}
  \begin{equation*} 
     \DslashL = D_W(0) + \mathcal{O}(a),
  \end{equation*}
  implying
  \begin{equation} \label{eq:WilsonDiscError}
      \DslashL + \frac{m_{0}}{a}I = D_W(m_0) + \mathcal{O}(a)
  \end{equation}
  for any mass parameter $m_0$.

  To construct $D_N$ we first note that any
  operator $\widehat{D}$ that is $\Gamma_5$-symmetric and
  fulfills~\eqref{GW:eq} can be parametrized by
    \begin{equation}\label{eq:NeubergerParametrization}
      a\widehat{D} = I+\Gamma_5 S, 
    \end{equation} with $S^{H} = S$ and $S^2 = I$. Both conditions are fulfilled for
    \begin{equation*}
      S = \Gamma_{5}D_{W}(m^{\ovl}_{0})\Big(D_{W}(m^{\ovl}_{0})^H
    (D_{W}(m^{\ovl}_{0})\Big)^{-\frac12}, \quad -m_{0}^\ovl \in
    \mathbb{R} \setminus \operatorname{spec}(D_{W}(0)).
    \end{equation*} 
    Using~\eqref{eq:WilsonDiscError} we obtain
    \begin{eqnarray*} 
      S &=& \Gamma_{5}\Big(\DslashL + \tfrac{m_{0}^\ovl}{a} I +
      \mathcal{O}(a)\Big)\Big(\big(\DslashL + \tfrac{m_{0}^\ovl}{a}I
      + \mathcal{O}(a)\big)^H \big(\DslashL + \tfrac{m_{0}^\ovl}{a}I
      + \mathcal{O}(a)\big)\Big)^{-\frac12} .  
    \end{eqnarray*}
    Since $\Dslash$ is anti-selfadjoint, we have $\DslashL^H = -\DslashL$ and thus 
    \begin{eqnarray*} \lefteqn{\Big(\big(\DslashL + \tfrac{m_{0}^\ovl}{a} I + \mathcal{O}(a)\big)^H
    \big(\DslashL + \tfrac{m_{0}^\ovl}{a} I + \mathcal{O}(a)\big)\Big)^{-\frac12} }
    &&  \\
    &=& \tfrac{a}{|m_0^\ovl|} \Big(\big(\tfrac{a}{m_{0}^\ovl}\DslashL + I + \mathcal{O}(a^2)\big)^H
    \big(\tfrac{a}{m_{0}^\ovl} \DslashL + I + \mathcal{O}(a^2)\big)\Big)^{-\frac12} \\
    &= & \tfrac{a}{|m_{0}^\ovl|} I + \mathcal{O}(a^2),
    \end{eqnarray*}
  which in turn yields
    \begin{equation}\label{eq:NeubergerOrder}
      S = \Gamma_{5}\Big(\frac{a}{|m_0^\ovl|}\DslashL + \operatorname{sign}(m_{0}^\ovl)I + \mathcal{O}(a^{2})\Big).
    \end{equation} 
    Combining~\eqref{eq:NeubergerOrder} with \eqref{eq:NeubergerParametrization} we find
    \begin{equation*}
      a\widehat{D} = I + \frac{a}{|m_0^\ovl|}\DslashL + \operatorname{sign}(m_{0}^\ovl)I + \mathcal{O}(a^{2}) .
    \end{equation*}
    so that for $m_{0}^\ovl < 0$ we have
    \begin{equation*}
      \widehat{D} = \frac{1}{|m_0^\ovl|}\DslashL + \mathcal{O}(a).
    \end{equation*} 
    This shows that $\widehat{D}$ is a first order discretization of
    $\Dslash$. For it to be stable 
    one has to choose $-2 < m_{0}^\ovl < 0$, a result for which we do
    not reproduce a proof here, referring to \cite{Neuberger1998141}
    instead.
    
    To conclude, note that $D_{N} = \widehat{D} +
      {\tfrac{\rho - 1}{a}I}$, so $\rho - 1$ sets the
    quark mass (see \eqref{Dirac_eq}) up to a renormalization factor. \qed
  \end{proof}

Using the Wilson-Dirac operator as the kernel in the overlap operator is 
the most popular choice, even though
other kernel operators have been investigated as well~\cite{deForcrand:2011ak}.
Neuberger's overlap operator has emerged as a popular scheme in lattice QCD over the years.\footnote{The domain wall
  discretization satisfies \eqref{GW:eq} approximately and, hence, has
  also been the focus of extensive research.} In the literature one
often writes
\begin{equation} \label{overlap_def:eq}
  D_N = \rho I + \Gamma_5 \sign\big(\Gamma_5 D_W(m_0^\ovl)\big)
\end{equation}
with $\sign$ denoting the matrix extension of the
sign function  
\[
\sign(z) = \left\{ \begin{array}{ll} +1 & \mbox{if } \Re(z) > 0 \\
  -1 & \mbox{if } \Re(z) < 0 \end{array} \right. .
\]
We note that $\sign(z)$ is undefined if $\Re(z) = 0$. 
Since $\Gamma_5D_W(m_0)$ is hermitian, see \eqref{gamma_5_symmetry:eq}, the matrix $\sign(\Gamma_5 D_W(m_0^\ovl))$ is hermitian, too. Since $\Gamma_5^2 = I$, we also see 
that the overlap operator satisfies the same $\Gamma_5$-symmetry as
its kernel $D_W$,
\begin{equation} \label{Gamma_5_symmetry_D_N:eq}
\big(\Gamma_5D_N\big)^H = \Gamma_5 D_N.
\end{equation}

We end this section with a characterization of the spectra of the
Wilson-Dirac and the overlap operator. 

\begin{lemma} \label{props_Wilson_Overlap:lem}
\begin{itemize}
\item[(i)] The spectrum of the Wilson-Dirac matrix $D_W(m_0)$ is symmetric to the real axis and to the vertical line $\Re(z) = \frac{m_0+4}{a}$, i.e., 
\[
  \lambda \in \spec\big(D_W(m_0)\big) \Rightarrow  \overline{\lambda}, \, {\textstyle 2 \frac{m_0+4}{a}-\lambda} \in \spec\big(D_W(m_0)\big).
\] 
\item[(ii)] The overlap operator $D_N$ is normal. Its spectrum is symmetric to the real axis and part of the circle with midpoint $\rho$ and radius 1, i.e.,
\[
  \lambda \in \spec\big(D_N\big) \Rightarrow  \overline{\lambda} \in \spec\big(D_N\big) \mbox{ and }
|\lambda-\rho| = 1.
\] 
\end{itemize}
\end{lemma}  
\begin{proof}
Recall that $\Gamma_5^H = \Gamma_5^{-1} = \Gamma_5$. If $D_W(m_0) x = \lambda x$, then by \eqref{gamma_5_symmetry:eq} we have $(\Gamma_5x)^HD_W = x^H(\Gamma_5D_W) = (\Gamma_5D_Wx)^H = \overline{\lambda} (\Gamma_5 x)^H$. This proves the first assertion in (i). For the second assertion,
consider a red-black ordering of the lattice sites. 
where all red sites appear before black sites.  Then the matrix $D_W(-\frac{4}{a})$ has the block structure
\[
  D_W(\textstyle{-\frac{4}{a}}) = \left( \begin{array}{cc} 0 & D_{rb} \\ D_{br} & 0 \end{array} \right).
\]
Thus, if $x = (x_r, x_b)$ is an eigenvector of $D_W(-\frac{4}{a})$ with eigenvalue $\mu$, then $x' = (x_r, -x_b)$ is an eigenvector of $D_W(-\frac{4}{a})$ with eigenvalue $-\mu$. Applying this result to 
$D_W(m_0)$ gives the second assertion in (i).

To prove (ii) we first remark that the sign function is its own inverse and that 
$\Gamma_5 D_W(m_0)$ is hermitian. This implies that $\sign(\Gamma_5D_W(m_0))$ is its own inverse and hermitian, thus unitary. Its product with the unitary matrix $\Gamma_5$ is unitary as well,  
implying that all its eigenvalues have modulus one. As a unitary matrix, this product is also normal. The term $\rho I$ in \eqref{overlap_def:eq} preserves normality and shifts the eigenvalues by $\rho$. 

It remains to show that $\spec(D_N)$ is symmetric with respect to the real axis, which follows from 
the $\Gamma_5$-symmetry \eqref{Gamma_5_symmetry_D_N:eq} of the overlap operator in the same manner as in (i). \qed
\end{proof}

For the purposes of illustration, Figure~\ref{spectra:fig} gives the spectra of the Wilson-Dirac operator and the overlap operator for a $4^4$ 
lattice. There, as everywhere else from now on, we set $a=1$ which is no restriction since $a^{-1}$ enters $D_W$ simply as a linear scaling. The 
matrix size is just $3,\!072$, so  all 
eigenvalues and the sign function can be computed with standard methods for full matrices. The choice for $m_0$ in the Wilson-Dirac matrix 
as a negative number such that the spectrum of $D_W$ lies in the right half plane with some eigenvalues being close to the imaginary axis is 
typical. The choice for $m_0$ when $D_W(m_0)$ appears in the kernel of 
the sign function is different (namely smaller, see Proposition~\ref{neuberger:prop}).
\begin{figure}
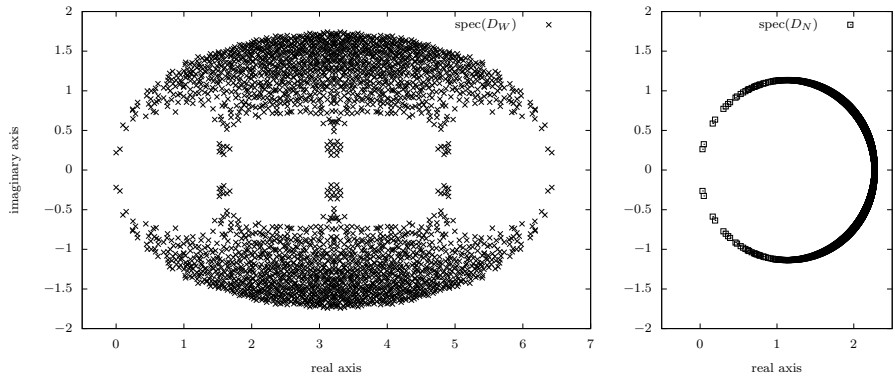

  \begin{minipage}[t]{0.65\textwidth}
    \centering\scalebox{0.61}{\input{./overlap_spec_wilson.tex}}
  \end{minipage}
  \hfill
  \begin{minipage}[t]{0.34\textwidth}
    \centering\scalebox{0.61}{\input{./overlap_spec_overlap.tex}}
  \end{minipage}
\caption{Typical spectra of the Wilson-Dirac and the overlap operator
for a $4^4$ lattice. 
}
\label{spectra:fig}
\end{figure}

\section{A Preconditioner Based on the Wilson-Dirac Operator} \label{prec:sec}

The spectral gaps to be observed as four discs with relatively few eigenvalues in the 
left part of Figure~\ref{spectra:fig} are typical for the spectrum of the Wilson-Dirac operator and become even more pronounced as lattice sizes are increased. In practice, the mass parameter $m_0$ that appears
in the definition of the kernel 
$D_W(m_0^\ovl)$ of the overlap operator is chosen such that the origin lies in the middle of the leftmost of these discs. For this choice of $m_0^\ovl$ we now motivate why the 
Wilson-Dirac operator $D_W(m_0^\prec)$ with adequately chosen mass $m_0^\prec$ provides a good preconditioner for the overlap operator.

To do so we investigate the connection of the spectrum of the overlap operator and the Wilson-Dirac operator in the special case that $D_W(0)$ is normal. This means that 
$D_W(0)$ is unitarily diagonalizable with possibly complex eigenvalues, i.e.,
\begin{equation} \label{eq_D_normal_dec}
D_W(0) = X \Lambda X^H, \mbox{ with } \Lambda \mbox{ diagonal and $X$ unitary.} 
\end{equation}
Trivially, then, $D_W(m_0)$ is normal for all mass parameters $m_0$ and
\begin{equation} \label{eq:D_m_normal}
D_W(m_0) = X(\Lambda + m_0I)X^H.
\end{equation}

To formulate the resulting non-trivial relation between the eigenvalues of 
$D_N$ and its kernel $D_W(m_0^\ovl)$ in the theorem below we use the notation $\csign(z)$ for a complex
number $z$ to denote its ``complex'' sign, i.e., 
\[
\csign(z) = z/|z| \mbox{ for } z \neq 0. 
\]
The theorem works with the singular value decomposition
$A = U \Sigma V^H$ of a matrix $A$ in which $U$ and $V$ are orthonormal, containing the left and right singular vectors as their columns, respectively, and $\Sigma$ is diagonal with non-negative diagonal elements, the 
singular values. The singular value decomposition is unique up to choices for the orthonormal basis of singular vectors belonging to the same singular value, i.e., up to transformations $U \to UQ, V \to VQ$ 
with $Q$ a unitary matrix commuting with $\Sigma$; cf.~\cite{GLMatrix1989}.   

\begin{theorem} \label{the:D_D_N_eigs}
Assume that $D_W(0)$ is normal, so that $D_W(m)$ is normal as well for all $m \in \mathC$, and let $X$ and $\Lambda$ be from \eqref{eq_D_normal_dec}. Then we have
\begin{equation} \label{D_N_normal:eq}
D_N = X\big(\rho I + \csign(\Lambda + m_0I)\big)X^H.
\end{equation}
\end{theorem}
\begin{proof}
Let
\begin{equation} \label{gamma_D_eigendecomposition:eq}
\Gamma_5 D_W(m) = W_{m} \Delta_{m} W_{m}^H \mbox{ with } \Delta_m \mbox{ diagonal}, W_m \mbox{ unitary},
\end{equation}
be the eigendecomposition of the hermitian matrix $\Gamma_5 D_W(m)$.
We have two different representations for the singular value decomposition of 
$\Gamma_5D_W(m)$,
\[
\begin{array}{rcll}
   \Gamma_5 D_W(m) &=& \big(\Gamma_5X\csign(\Lambda+mI)\big) \cdot |\Lambda +mI| \cdot  X^H  
 &\quad \mbox{(from \eqref{eq:D_m_normal})}\; , \\
  \Gamma_5 D_W(m) &=& \big( W_m\sign(\Delta_m)\big) \cdot |\Delta_m| \cdot W_m^H  
   &\quad \mbox{(from \eqref{gamma_D_eigendecomposition:eq})} \; . 
\end{array}
\]
Thus, there exists a unitary matrix $Q$ such that
\begin{equation} \label{themess:eq}
W_m = XQ \mbox{ and } W_m \sign(\Delta_m) = \Gamma_5 X \csign(\Lambda + mI) Q.
\end{equation}
Using the definition of $D_N$ in \eqref{overlap_def:eq}, the relations \eqref{themess:eq} give
\begin{eqnarray*}
D_N &=& \rho I + \Gamma_5 \sign(\Gamma_5 D_m) \\
   &=& \rho I + \Gamma_5 W_m \sign(\Delta_m) W_m^H \\
&=& \rho I + \Gamma_5 \Gamma_5 X \csign (\Lambda + mI)Q (VQ)^X \\
&=& X(\rho I + \csign(\Lambda +mI)X^H. 
\end{eqnarray*}
\qed
\end{proof}

We remark in passing that as an implicit consequence of the proof above we have that the eigenvectors 
of $\Gamma_5 D_W(m) = \Gamma_5 D_W(0) + m\Gamma_5$ do not depend on
$m$. Thus if $D_W$ is normal, 
$\Gamma_5$ and $\Gamma_5D_W$ admit a basis of common eigenvectors. 

The result in \eqref{D_N_normal:eq} implies that $D_N=\rho I + \Gamma_5\sign(\Gamma_5D_W(m_0^\ovl))$ and $D_W(0)$ 
share the same eigenvectors and that 
\[
\spec(D_N) = \{ \rho + \csign(\lambda +m^\ovl_0), \lambda \in \spec(D_W(0))\}.
\]
Taking $D_W(m_0^\prec)$ as a preconditioner for $D_N$, we would like eigenvalues of $D_N$ which are small in 
modulus to be mapped to eigenvalues close to 1 in the preconditioned matrix $D_ND_W(m_0^\prec)^{-1}$.
Since $D_W(m_0^\prec)$ and $D_N$ share the same eigenvectors, the spectrum of the preconditioned matrix 
is
\[
\spec\big(D_N D_W(m_0^\prec)^{-1}\big) = \Big\{ {\frac{\rho + \csign(\lambda +m_0^\ovl)}{\lambda + m_0^\prec}}, \lambda \in \spec(D_W(0) \Big\}.
\]
For $\omega > 0$ and $m_0^\prec = \omega\rho+m_0^\ovl$, the mapping
\[
g: \mathC \to \mathC, z \mapsto {\frac{\rho + \csign(z +m_0^\ovl)}{z + m_0^\prec}}
\]
sends $C(-m_0^\ovl,\omega)$, the circle with center $-m_0^\ovl$ and radius $\omega$, to one single value $\frac{1}{\omega}$.
We thus expect $D_W(m_0^\prec)$ to be a good preconditioner if we
choose $m_0^\prec$ in such a manner
that the small eigenvalues of $D_W(m_0^\prec)$ lie close to $C(-m_0^\ovl,\omega)$. 
Let $\sigma_{\min} > 0$ denote the smallest real part of all eigenvalues of
$D_W(0)$. Assuming for the moment that $\sigma_{\min}$ is actually an eigenvalue, this eigenvalue will lie exactly 
on $C(-m_0^\ovl,\omega)$ if we have 
\begin{equation} \label{eq:default_m}
\omega = \omega^{\default} := -m_0^\ovl- \sigma_{\min} \mbox{ and thus } m_0^\prec = m_0^{\default} := \omega^{\default} \rho + m_0^\ovl.
\end{equation}

For physically relevant parameters, $\omega^\default$ is close to 1.
We will take $m_0^{\default}$ from \eqref{eq:default_m} as our default choice for the mass parameter when preconditioning
with the Wilson-Dirac operator, although a slightly larger value for $\omega$ might 
appear adequate in situations where the eigenvalues with smallest real part come as a complex conjugate 
pair with non-zero imaginary part.

Although $D_W(0)$ is non-normal in physically relevant situations, we expect the above reasoning 
to also lead to an effective  Wilson-Dirac preconditioner in these settings, and particularly so when the deviation of $D_W(0)$
from normality, as measured in some suitable norm of $D_W^HD_W - D_WD_W^H$, becomes small. This is so, e.g., 
when the lattice spacing is decreased while keeping the physical volume constant, i.e., in 
the ``continuum limit'', since the Wilson-Dirac operator then approaches the continuous Dirac operator which is normal. Moreover, 
as we will show in 
section~\ref{smearing:sec}, when {\em smearing} techniques are applied to a given gauge 
configuration $U_\mu(x)$,  the deviation of $D_W(0)$ from normality is also decreased.
Figure~\ref{fig:spectra_prec} shows the spectrum for the preconditioned matrix with the choice \eqref{eq:default_m}
for $m_0^\prec$ for the same $4^4$  configuration as in Figure~\ref{spectra:fig}. The matrices in 
these tests are not normal, nonetheless the spectrum of the preconditioned matrix tends to 
concentrate around 0.7. 

\begin{figure}
  \centering \scalebox{0.7}{\input{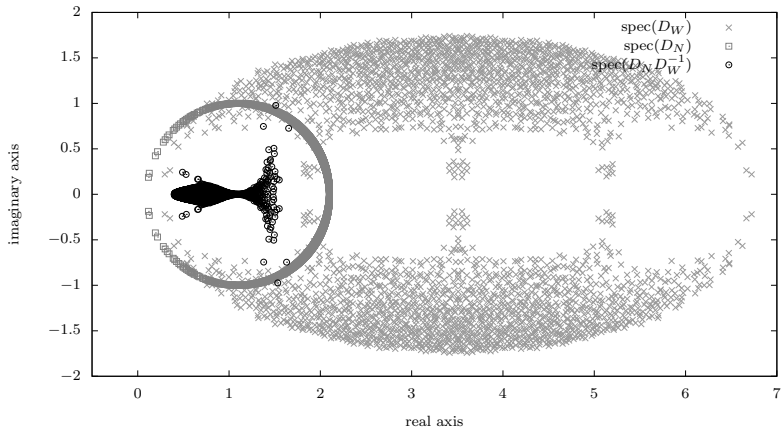}}
  \caption{Spectra for a configuration of size $4^4$}
  \label{fig:spectra_prec}
\end{figure}

In the normal case, the singular values are the absolute values of the eigenvalues, and the singular vectors 
are intimately related to the eigenvectors. This relation was crucial to the proof of Theorem~\ref{the:D_D_N_eigs}. In the non-normal case, the relation \eqref{D_N_normal:eq}, which uses the eigenvectors of $D_W(0)$, does not hold. For the sake of completeness we give, for the general, non-normal case, the following result which links the overlap operator to the singular value decomposition of its kernel $D_W(m)$.

\begin{lemma} \label{lem:svd_kernel}
Let $\Gamma_5 D_W(m) = W_m \Delta_m W_m^H$ denote an eigendecomposition of the 
hermitian matrix
$\Gamma_5D_W(m)$, where $\Delta_m$ is real and diagonal and $W_m$ is unitary. 
Then 
\begin{itemize}
   \item[(i)] A singular value decomposition of $D_W(m)$ is given as 
    \[
       D_W(m) = U_m \Sigma_m V_m^H \mbox{ with } V_m = W_m, \Sigma_m = |\Delta_m|, U_m = \gamma_5  
     W_m \sign(\Delta_m).
  \]
     \item[(ii)]  The overlap operator with kernel $D_W(m)$ is given as
\[
    D_N = \rho I + \Gamma_5 \sign\big(\Gamma_5D_W(m) \big) = \rho I + U_mV_m^H.
\]
\end{itemize}
\end{lemma}
\begin{proof}
Since $\Gamma_5^{-1} = \Gamma_5$, we have the factorization $D_W(m) = \Gamma_5 W_m \Delta_m W_m^H = 
\Gamma_5 W_m \sign(\Delta_m) |\Delta_m| W_m^H$, in which $\Gamma_5 W_m \sign(\Delta_m)$ and $W_m$
are unitary and $|\Delta_m|$ is diagonal and non-negative. This proves (i). To show (ii), just observe that for the hermitian matrix $\Gamma_5 D_W(m)$ we have $\sign(\Gamma_5 D_W(m)) = W_m \sign(\Delta_m)W_m^H$ and use (i). 
\end{proof}

\section{Smearing and Normality} \label{smearing:sec}

To measure the deviation from normality of $D_W$ we 
now look at the Frobenius norm of $D_W^HD_W-D_WD_W^H$. 
We show that this measure can be fully expressed 
in terms of the pure gauge action, defined as a sum of path-ordered products of link variables, the {\em plaquettes}, to be defined 
in detail below. Based on this connection 
we then explain that ``stout'' smearing~\cite{Morningstar:2003gk}, a 
modification of the gauge links by averaging with neighboring links, has the effect 
of reducing the non-normality of $D_W$, among its other physical
benefits. This result indicates that preconditioning with the Wilson-Dirac operator and using the choice \eqref{eq:default_m} for $m^\ovl$ is increasingly better motivated as more smearing steps are applied. This observation will be substantiated by numerical experiments in section~\ref{numerics:sec}.

%

\begin{definition} Given a configuration of gauge links $\{U_\mu(x)\}$, the 
{\em plaquette} $\pla_x^{\mu,\nu}$ at lattice point $x$ is defined as  
\begin{equation} \label{plaquette_def1:eq}
\pla_x^{\mu,\nu} = U_\nu(x)U_\mu(x+\hat{\nu})U^H_\nu(x+\hat{\mu})U^H_\mu(x).
\end{equation}
\end{definition}
A plaquette thus is the product of all coupling matrices along a cycle of length 4
on the torus, such cycles being squares in a $(\mu,\nu)$-plane
\[
Q_x^{\mu,\nu} \mathrel{\widehat{=}} \Q{1}{1} \enspace .
\]
Similarly, the plaquettes in the other quadrants are defined as
\begin{equation} \label{plaquette_def2:eq}
Q_x^{\mu,-\nu} \mathrel{\widehat{=}}  \Q{1}{-1}, \quad
Q_x^{-\mu,\nu}  \mathrel{\widehat{=}}  \Q{-1}{1}, \quad
Q_x^{-\mu,-\nu}  \mathrel{\widehat{=}}  \Q{-1}{-1} \ .
\end{equation}
Note that on each cycle of length four there are four plaquettes which are conjugates of each other.
They are defined as the products of the gauge links along that cycle with different starting sites,
so that we have, e.g., $Q_{x+\hat{\mu}}^{-\mu,\nu} =
U^H_{\mu}(x)Q_x^{\mu,\nu}U_\mu(x)$, etc.

The deviation of the plaquettes from the identity is a measure for
the non-normality of $D$ as determined by the following proposition. Its proof is obtained by simple, though technical, algebra which we summarize in the appendix. 

\begin{proposition}\label{Fnorm:prop} The Frobenius norm of $D_W^HD_W-D_WD_W^H$ fulfills
\begin{equation}\label{eq:deviationfromnormality}
\| D^H_WD_W-D_WD_W^H \|_F^2 =  16 \sum_x \sum_{\mu < \nu} \Re(\tr{I-\pla_x^{\mu,\nu}})  
\end{equation}
where the first sum is to be taken over all lattice sites $x$ and  
$\sum_{\mu < \nu}$ is a shorthand for $\sum_{\mu = 0}^3 \sum_{\nu = \mu+1}^3$. 
\end{proposition}

As a consequence of Proposition~\ref{Fnorm:prop} we conclude that $D_W$ is normal in the case of the {\em free theory},
i.e., when all links $U_\mu(x)$ are equal to the identity or when $U_\mu(x) = U(x)U^H(x+\hat{\mu})$ for a collection of
$SU(3)$-matrices $U(x)$ on the lattice sites $x$. 
For physically relevant configurations, however, $D_W$ is non-normal. The quantity 
\[ 
\sum_x \sum_{\mu < \nu} \mbox{Re}(\tr{I-\pla_x^{\mu,\nu}})
\]
is known as the {\em Wilson gauge action}\footnote{To represent a physically meaningful quantity, the Wilson gauge action is usually 
scaled with a scalar factor. This is not relevant in the present context.} $S_W(\mathcal{U})$ of the gauge field $\mathcal{U} = \{U_\mu(x)\}$.  

Smearing techniques for averaging neighboring gauge links have been studied extensively in 
lattice QCD simulations. Their use in physics is motivated by the goal to reduce ``cut-off 
effects'' related to localized eigenvectors with near zero eigenvalues.
We now explain why ``stout'' smearing~\cite{Morningstar:2003gk} reduces 
the Wilson gauge action and thus drives the Wilson-Dirac operator 
towards normality. Other smearing techniques like APE~\cite{Albanese:1987ds},
HYP~\cite{Hasenfratz:2007rf} and HEX~\cite{Capitani:2006ni} have similar effects. 

Given a gauge field $ \mathcal{U}$,
stout smearing modifies the gauge links according to
\begin{equation} \label{stout:eq}
U_\mu(x) \to \tilde{U}_\mu(x) = \mathrm{e}^{\epsilon Z^{\mathcal{U}}_\mu(x)} U_\mu(x)\,
\end{equation}
where the parameter $\epsilon$ is a small positive number and 
\begin{eqnarray}
 Z^{\mathcal{U}}_\mu(x) 
               & = & -\frac{1}{2}(M_\mu(x)-M^H_\mu(x)) + \frac{1}{6}\tr{M_\mu(x)-M^H_\mu(x)}\,, \label{Zdef:eq} 
\end{eqnarray}
where
\begin{eqnarray}
              M_\mu(x) &=&  \sum_{\nu=0,\nu \neq \mu}^3 Q_x^{\mu,\nu} + Q_x^{\mu,-\nu}. \nonumber
\end{eqnarray}
Note the dependence of $Z^{\mathcal{U}}_\mu(x)$ on local plaquettes associated with $x$. 

The following result from \cite{Luscher:2010iy,Luscher:2009eq} relates the {\em Wilson flow} $\mathcal{V}(\tau) = \{V_\mu(x,\tau): x \in \mathcal{L}, \mu=0,\ldots,3\}$
defined as the solution of the initial value problem 
\begin{equation}\label{eq:floweq}
 \frac{\partial}{\partial \tau}V_\mu(x,\tau) = - \left\{ \dxmu \Sw(\mathcal{V}(\tau))\right\} V_\mu(x,\tau)\,, \quad V_\mu(x,0) = U_\mu(x)\,,
\end{equation}
to stout smearing.  Here 
$V_\mu(x,\tau) \in \SU(3)$, and $\dxmu$ is the canonical differential operator
with respect to the link variable $V_\mu(x,\tau)$ which takes values in $\su(3)$, the algebra of $\SU(3)$. 

\begin{theorem} Let $\mathcal{V}(\tau)$ be the solution of \eqref{eq:floweq}. Then
\begin{itemize}
\item[(i)] $\mathcal{V}(\tau)$ is unique for all $\mathcal{V}(0)$ and all $\tau\in(-\infty,\infty)$ and differentiable
with respect to $\tau$ and $\mathcal{V}(0)$. 
\item[(ii)] $S_W(\mathcal{V}(\tau))$ is monotonically decreasing as a function of $\tau$.
\item[(iii)] One step of Lie-Euler integration with step size $\epsilon$  for \eqref{eq:floweq}, starting at $\tau = 0$, gives the approximation $\widetilde{\mathcal{V}}(\epsilon) = \{\widetilde{V}_\mu(x,\epsilon)\}$ for $\mathcal{V}(\epsilon)$ with
\[
\widetilde{V}_\mu(x,\epsilon) = \mathrm{e}^{\epsilon Z^{\mathcal{U}}_\mu(x)} U_\mu(x),
\]
with $Z^{\mathcal{U}}_\mu(x)$ from \eqref{Zdef:eq}
\end{itemize}
\end{theorem}
We refer to \cite{Luscher:2010iy,Luscher:2009eq} and also \cite{Bonati14} for details of the proof for (i) and (ii). 
It is noted in~\cite{Bonati14} that the solution of~\eqref{eq:floweq} moves the gauge configuration
along the steepest descent in configuration space and thus actually minimizes
the action locally. Part (iii) follows directly by applying the Lie-Euler scheme; cf.~\cite{HaLuWa06}.

The theorem implies that one Lie-Euler step is equivalent to a step of
stout smearing, with the exception that in stout smearing links are
updated sequentially instead of in parallel. And since the Wilson
action decreases along the exact solution of \eqref{eq:floweq}, we can
expect it to also decrease for its Lie-Euler approximation, at least
when $\epsilon$ is sufficiently small. 

\begin{figure}
  \centerline{\scalebox{0.7}{\input{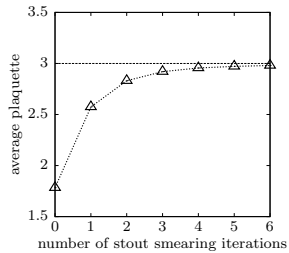}}}
  \caption{Illustration of the effect of stout smearing on the average plaquette value~\eqref{eq:avgplaquette}.\label{fig:smearingVSavgplaquette}}
\end{figure}

In Figure~\ref{fig:smearingVSavgplaquette} we illustrate the
relation between iterations of stout smearing and the average
plaquette value $Q_{\mathit{avg}}$ for configuration~\ref{JF_32_32}
(cf.~Table~\ref{table:allconfs}). The average plaquette value is defined by
\begin{equation}\label{eq:avgplaquette}
Q_{\mathit{avg}} = N_Q^{-1}\sum_x \sum_{\mu < \nu}
    \mathit{Re}(\tr{\pla_x^{\mu,\nu}}),
\end{equation} where $N_Q$ denotes the total number of
plaquettes. In terms of $Q_{\mathit{avg}}$~\eqref{eq:deviationfromnormality}
simplifies to 
\[
\norm[F]{D_{W}^HD_W - D_WD_W^H} =
16N_Q(3-Q_{\mathit{avg}}).
\] Figure~\ref{fig:smearingVSavgplaquette} shows that the Wilson action decreases rapidly in the first iterations of
stout smearing.

To conclude this section we note that there are several works
relating the spectral structure and the distribution of plaquette
values. For example, it has been shown in~\cite{Neuberger:1999pz} that the
size of the spectral gap around $0$ of $\Gamma_5 D_W$ is related to
$\Re(\tr{I-\pla_x^{\mu,\nu}})$ being larger than a certain threshold
for all plaquettes $\pla_x^{\mu,\nu}$. Other studies consider the 
connection between fluctuations of the plaquette value and localized zero modes,
see~\cite{Berruto:2000fx,Negele:1998ev,Niedermayer:1998bi}, and the
influence of smearing on these modes~\cite{Hasenfratz:2007iv}.



\section{Numerical Results} \label{numerics:sec}

\begin{table}
\centering\scalebox{0.9}{\tabcolsep=0.16cm\begin{tabular}{ccccccc}
\toprule
ID                                   & lattice size       & kernel mass                    & default            & smearing                            & provided by                              &\\
                                     & $N_t \times N_s^3$ & $m_0^\ovl$                     & overlap mass $\mu$ & $s$                                 &                                          &\\
\midrule
\ref{JF_32_32}\conflabel{JF_32_32}   & $32 \times 32^3$   & $ -1-\frac{3}{4}\sigma_{\min}$ & $0.0150000$        & $\{0,\ldots,6\}$-stout~\cite{Morningstar:2003gk} & generated from~\cite{DelDebbio:2006cn,DelDebbio:2007pz} &\\
\ref{BMW_32_32}\conflabel{BMW_32_32} & $32 \times 32^3$   & $-1.3$                         & $0.0135778$        & $3$HEX~\cite{Capitani:2006ni}                   & BMW-c, based on~\cite{Borsanyi:2012xf,Toth:2014uqa}                  &\\
\bottomrule
\end{tabular}}
\caption{Configurations used together with their parameters.
  See the references for details about their generation.}
\label{table:allconfs}
\end{table}

In physical simulations, gauge fields are generated via a stochastic process and by fixing physical parameters.
The term {\em configuration} designates a gauge field together with the information about its generation and its physical parameters.
In this section we report numerical results obtained on relatively large configurations used in current simulations involving the overlap operator, detailed in Table~\ref{table:allconfs}. The configurations with ID~\ref{JF_32_32} are available with 
different numbers $s=0,\ldots,6$ of stout smearing steps applied. Note that $s$ influences $\sigma_{\min}$, the smallest real part of all eigenvalues of $D_W(0)$. The given choice for $m_0^\ovl$ as a function of $\sigma_{\min}$, used in $D_N = \rho I + \Gamma_5 \sign(\Gamma_5 D_W(m_0^\ovl))$  places the middle of the first `hole' in the spectrum of $D_W(m_0^\ovl)$ 
to be at the origin. The configuration with ID~\ref{BMW_32_32} was
obtained using 3 steps of HEX smearing in a simulation similar in spirit to \cite{Borsanyi:2012xf,Toth:2014uqa} with its physical results not yet published. The value $m_0^\ovl = -1.3$ is
the one used in the simulation. The middle of the first `hole' in $D_W(m_0^\ovl)$ is thus close to but not exactly at the origin.    
To be in line with the conventions from~\cite{Borsanyi:2012xf}, e.g., we express the parameter $\rho\geq 1$ used in the overlap operator
$D_N$ as
\[
\rho = \frac{-\mu/2 + m_0^{\ovl}}{\mu/2 + m_0^{\ovl}},
\] 
where $\mu >0$ is yet another, ``overlap'' mass parameter.
In our experiments, we will frequently consider a whole range for $\mu$ rather than just the default value from Table~\ref{table:allconfs}. 
The default value for $\mu$ is chosen such that it fits to other physically interpretable properties of the respective configurations like, 
e.g., the pion mass $m_\pi$. For both sets of configurations used, $m_\pi$ is approximately twice as large than the value observed in nature, 
and the ultimate goal is to drive $m_\pi$ to its physical value, which very substantially increases the cost for generating the respective configurations. We would then use smaller values for $\mu$, and the results of our experiments for such smaller $\mu$ hint at how the preconditioning will perform in future simulations at physical parameter values.
Note that smaller values for $\mu$ make $\rho$ become closer to 1, so $D_N$ becomes more ill-conditioned.

All results were obtained on the Juropa machine at J\"ulich
Supercomputing Centre, a cluster with $2,\!208$ compute nodes, each
with two Intel Xeon X5570 (Nehalem-EP) quad-core processors \cite{IntelXeonX5570,wwwJUROPA}.
This machine provides a maximum of $8,\!192$ cores for a single job from which we always use $1,\!024$ in our experiments.
For compilation we used the \texttt{icc}-compiler with the optimization flags
\texttt{-O3}, \texttt{-ipo}, \texttt{-axSSE4.2} and \texttt{-m64}.
In all tests, our code ran with roughly $2$ Gflop/s per core which accounts to $8-9\%$ peak performance.
The multigrid solver used to precondition with $D_W(m_0^\prec)$ (see below) performs at roughly $10\%$ peak.

\subsection{Accuracy of the preconditioner and influence of $m_0^\prec$}

In a first series of experiments, we solve the system 
\begin{equation} \label{linsys:eq}
D_N \psi = \eta
\end{equation}
on the one hand without any preconditioning, using GMRES(100), i.e., restarted GMRES with a cycle length of 100.
On the other hand, we solve the same system using $D_W^{-1}$ as a (right) preconditioner. To solve the respective 
linear systems with $D_W$ we use the domain decomposition based adaptive algebraic multigrid method 
(DD-$\alpha$AMG) presented in \cite{FroKaKrLeRo13}. Any other efficient solver for Wilson-Dirac equations as, e.g., 
the ``AMG'' solver developed in~\cite{MGClark2010_1,MGClark2007,Luescher2007,MGClark2010_2} could
be used as well. In our approach, preconditioning 
is done by iterating with DD-$\alpha$AMG until the relative residual is below a prescribed bound $\epsilon^{\prec}$. 
Without going into detail, let us mention that DD-$\alpha$AMG uses a red-black multiplicative Schwarz method as 
its smoother and that it needs a relatively costly, adaptive setup-phase in which restriction and prolongation 
operators---and with them the coarse grid systems---are constructed. We refer to \cite{FroKaKrLeRo13} for further reading.
The setup has to be done only once for a given Wilson-Dirac matrix $D_W$, so its cost becomes negligible when 
using DD-$\alpha$AMG as a preconditioner in a significant number of GMRES iterations.\footnote{In all our experiments, the setup 
never exceeded 2\% of the total execution time, so we do not report timings for the setup.}
We use GMRES with odd-even preconditioning~\cite{MGClark2010_2} as a solver for the coarsest system.
The whole DD-$\alpha$AMG preconditioning iteration is non-stationary which has to be accounted for by using {\em flexible}
restarted GMRES (FGMRES) \cite{Saad:2003:IMS:829576} to solve \eqref{linsys:eq} instead of GMRES. The restart length for 
FGMRES is again $100$.

\begin{figure} [thb]
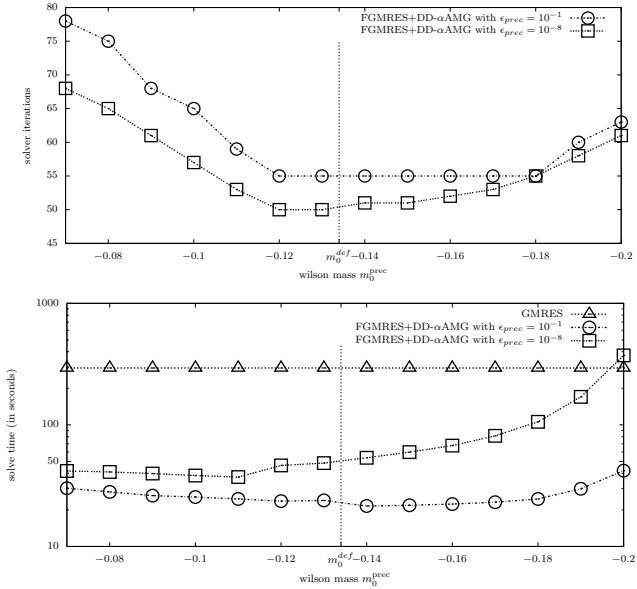

  \hspace*{0.15\textwidth}\scalebox{0.485}{\input{./scan_m0w_32s3_iter}}
  \hspace*{0.132\textwidth}\scalebox{0.50}{\input{./scan_m0w_32s3}}
  \caption{Preconditioner efficiency as a function of $m_0^\prec$ for two accuracies for the DD-$\alpha$AMG solver (configuration ID~\ref{JF_32_32}, $s=3$). Top: number of iterations, bottom: execution times.}
  \label{fig:prec_eff_m0w}
\end{figure}

Figure~\ref{fig:prec_eff_m0w} presents results 
for configuration ID~\ref{JF_32_32} with $s=3$ stout smearing steps and the default overlap mass $\mu$ from Table~\ref{table:allconfs}.
We scanned the values for $m_0^\prec$ in 
steps of $0.01$ and report the number of iterations necessary to reduce the initial residual by a factor of $10^{-8}$ for each of these values.
We chose two different values $\epsilon^\prec$ for the residual reduction required in the DD-$\alpha$AMG iteration in the preconditioning.
The choice $\epsilon^\prec = 10^{-8}$ asks for a relatively accurate solution of the systems with $D_W(m_0^\prec)$, whereas the choice 
$\epsilon^\prec = 10^{-1}$ requires an only quite low accuracy and thus only a few iterations of DD-$\alpha$AMG.
The upper part of Figure~\ref{fig:prec_eff_m0w} shows that there is a dependence of the number of FGMRES iterations on $m_0^\prec$, while at the 
same time there is a fairly large interval around the optimal value for $m_0^\prec$ in which the number of iterations required is not more than 20\% larger than the minimum. 
These observations hold for both accuracy requirements for the DD-$\alpha$AMG solver, $\epsilon^\prec = 10^{-8}$ and $\epsilon^\prec = 10^{-1}$. 
The number of iterations needed without preconditioning was 973.

The lower part of Figure~\ref{fig:prec_eff_m0w} shows that similar observations hold for the execution times.
However, the smaller iteration numbers obtained with $\epsilon^\prec = 10^{-8}$ do not translate into smaller execution times, 
since the time for each DD-$\alpha$AMG solve in the preconditioning is substantially higher as for $\epsilon^\prec = 10^{-1}$.
This turned out to hold in all our experiments, so from now on we invariably report results for $\epsilon^\prec = 10^{-1}$.  
We also observe that the value of $m_0^\default$ from \eqref{eq:default_m} lies within an interval in which iteration numbers and execution times (for both values for $\epsilon^\prec$) are quite close to the optimum. The execution time without preconditioning 
was 294s.

\begin{figure}[htb]
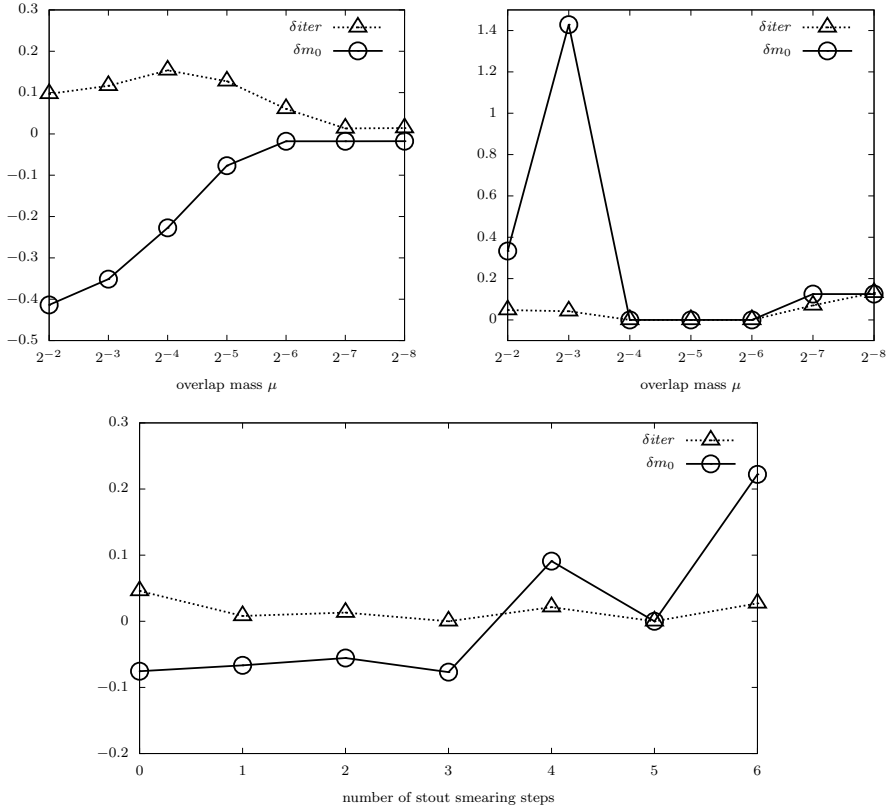

 
  \begin{minipage}[t]{0.48\textwidth}
  \centering \scalebox{0.68}{\input{./scan_rho_m0w_diff_estimate_32s0_iter}}
  \end{minipage} \hfill
   \begin{minipage}[t]{0.48\textwidth}
  \centering \scalebox{0.68}{\input{./scan_rho_m0w_diff_estimate_32s3_iter}}
  \end{minipage} 
   \centering \scalebox{0.68}{\input{./smearing_scan_m0w_diff_estimate}}
  \caption{Quality of $m_0^\default$ without smearing (top left), with $s=3$ steps of stout smearing (top right), and for $s=0,\ldots,6$ steps of 
  stout smearing at fixed $\mu$ (bottom), configuration ID~\ref{JF_32_32}.}
  \label{fig:estimate_s0_s3}
\end{figure}

Figure~\ref{fig:estimate_s0_s3} reports results which show that the default value $m_0^\default$ is a fairly good choice in general.
For two different configurations (no smearing and 3 steps of stout smearing) and a whole range of overlap masses $\mu$, the plots at the top give
the relative difference $\delta m_0 = (m_0^{\mathit{opt}} - m_0^\default)/m_0^\default$ of the optimal value $m_0^{\mathit{opt}}$ for $m_0^\prec$ and its default value from 
\eqref{eq:default_m} as well as the similarly defined relative difference $\delta \mathrm{iter}$ of the corresponding iteration numbers. These 
results show that the iteration count for the default value $m_0^\default$ is never more than 15\% off the best possible iteration count. 
The plot at the bottom backs these findings. We further scanned a whole range of smearing steps $s$ at the default value for $\mu$ 
from Table~\ref{table:allconfs}, and the number of 
iterations with $m_0^\default$ is never more than $5\%$ off the optimal value. The large values for $\delta m_0$ in the top right plot for 
$\mu = 2^{-3}$ are to be attributed to the fact that the denominator in the definition of $\delta m_0$, i.e., 
$m_0^\default$ is almost zero in this case.

These results suggest that \eqref{eq:default_m} is indeed a good choice for $m_0^\prec$. 
However, $\sigma_{\min}$ needed to compute $m_0^\default$  from
\eqref{eq:default_m} is not necessarily known a priori, and it may be
more efficient to approximate the optimal value for $m_0$ ``on the
fly'' by changing its value from one preconditioned FGMRES iteration
to the next. 

In order to minimize the influence of the choice of
$m_0^\prec$ on the aspects discussed in the following sections we will
always use the optimal $m_0^\prec$, computed to a precision of $.01$ by scanning
the range $[-\widetilde{\sigma}_{\min},0]$, where
$\widetilde{\sigma}_{\min}$ is a rough guess at $\sigma_{\min}$ which
fulfills $\widetilde{\sigma}_{\min} > \sigma_{\min}$. This guess can be
easily obtained by a fixed number of power iterations to get an
approximation for the largest real part $\widetilde{\sigma}_{\max}$ of
an eigenvalue of $D$ and then using the symmetry
of the spectrum to obtain $\widetilde{\sigma}_{\min}$ by rounding
$8-\widetilde{\sigma}_{\max}$ to the first digit.

\subsection{Quality and cost of the preconditioner}

We proceed to compare in more detail preconditioned FGMRES($100$) with
unpreconditioned GMRES($100$) in terms of the iteration count. As
before, the iterations were stopped when the initial residual was
reduced by a factor of at least $10^{-8}$.


\begin{figure}[htb]
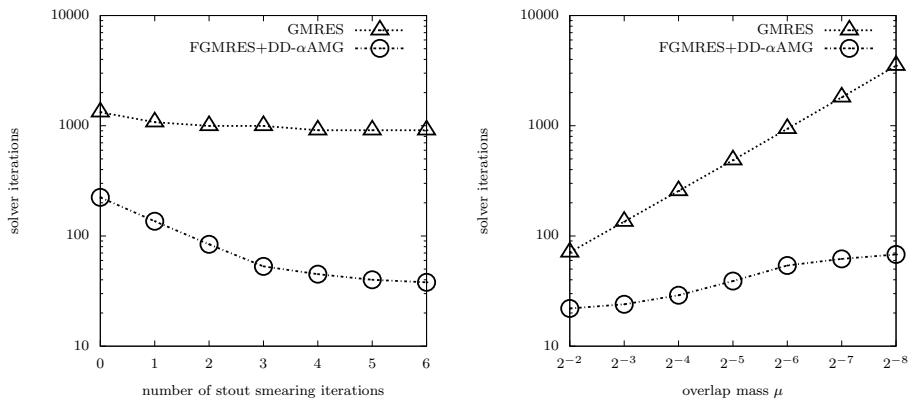

  \begin{minipage}[t]{0.48\textwidth}
  \centering \scalebox{0.68}{\input{./scan_s_32}}
  \end{minipage} \hfill
  \begin{minipage}[t]{0.48\textwidth}
  \centering \scalebox{0.68}{\input{./scan_rho_32s3}}
  \end{minipage}
  \caption{Comparison of preconditioned FGMRES(100) with unpreconditioned GMRES(100) (configuration ID~\ref{JF_32_32}). Left: dependence on the 
  number of stout smearing steps $s$ for default value for  $\mu$, cf.\ Table~\ref{table:allconfs}. Right: dependence on the overlap mass $\mu$ for $s=3$.}
  \label{fig:gmres_comparison}
\end{figure}

Figure~\ref{fig:gmres_comparison} gives this comparison, once as a function of the non-normality of the configuration, i.e., the number $s$ of stout smearing steps applied, and once as a function of the overlap mass $\mu$.
We see that for the default value of $\mu$ from Table~\ref{table:allconfs}, the quality of the preconditioner increases with the number $s$ of stout smearing steps, 
ranging from a factor of approximately $5$ for $s=0$ over $12$ for $s=3$ up to $25$ for $s=6$.  We also see that the quality of the preconditioner increases as $\mu$ becomes smaller, i.e., when $D_N$ becomes more ill-conditioned. 

From the practical side, a comparison of the execution times is more important than comparing iteration numbers. Before giving timings, we
have to discuss relevant aspects of the implementation in some detail.

Each iteration in GMRES or preconditioned FGMRES for \eqref{linsys:eq} requires one matrix vector multiplication with $D_N = \rho I + \Gamma_5\sign(\Gamma_5D_W)$. The matrix $D_N$ is not given explicitly as it would be a full, very large matrix despite $\Gamma_5 D_W$ 
being sparse. Therefore, a matrix vector multiplication $D_N\chi$ is obtained via an additional ``sign function iteration'' which approximates 
$\sign(\Gamma_5 D_W)\chi$ as part of the computation of $D_N\chi$. For this sign function iteration we use the restarted Krylov subspace method proposed 
recently in \cite{FrGuSc14b,FrGuSc14} which allows for thick restarts of the Arnoldi process and has proven to be among the most efficient methods to approximate $\sign(\Gamma_5 D_W)\chi$. The sign function iteration then still represents the by far most expensive part of the overall computation. 

A first approach to reduce this cost, see \cite{Cu05a}, is to use relaxation in the sense that one lowers the (relative) accuracy 
$\varepsilon_{\sign}$ of the 
approximation as the outer (F)GMRES iteration proceeds. The theoretical analysis of inexact Krylov subspace methods in \cite{SiSz,SlevdE}
shows that the relative accuracy of the approximation to the matrix-vector product at iteration $k$ should be in the order of 
$\epsilon/ \| r_k \|$ (with $r_k$ the (F)GMRES residual at iteration $k$) to achieve that at the end of the (F)GMRES iteration the initial residual be decreased by a factor of $\epsilon$. We used this relaxation strategy in our experiments.

A second commonly used approach, see e.g.~\cite{Edwards:1998yw,Giusti:2002sm,vdE02}, to reduce the cost of the sign function iteration is deflation. In this approach the $k$ smallest in modulus eigenvalues $\lambda_1,\ldots,\lambda_k$ and their normalized eigenvectors $\xi_1,\ldots,\xi_k$ are precomputed once. With $\Xi = [\xi_1|\ldots|\xi_k]$ 
and $\Pi = I-\Xi\Xi^H$ the orthogonal projector on the complement of these eigenvectors, $\sign(\Gamma_5D_W)\chi$ is given as
\[
\sign(\Gamma_5D_W)\chi = \sum_{i=1}^k \sign(\lambda_i) (\xi^H\chi) \xi_i + \sign(\Gamma_5D_W) \Pi \chi.
\]
The first term on the right side can be computed explicitly and the second term is now easier to approximate with the sign function iteration,
since the $k$ eigenvalues closest to the singularity of $\sign(\cdot)$ are effectively eliminated via $\Pi$.

\begin{table}[htb]
\centering\scalebox{0.9}{\begin{tabular}{llcc}
\toprule
                        & parameter                                &   notation                 & default    \\
\midrule
(F)GMRES${}^{dp}$       & required reduction of initial residual   & $\varepsilon_{\outer}$     & $10^{-8}$  \\
                        & relaxation strategy                      & $\varepsilon_{\sign}$      & $\frac{\varepsilon_{\outer}}{\| r_k\|}\cdot 10^{-2}$                                                                 \\
                        & restart length for FGMRES                & $m_{\restart}$             & $100$                   \\
\midrule
DD-$\alpha$AMG${}^{sp}$ & required reduction of initial residual   & $\varepsilon_{\prec}$      & $10^{-1}$  \\
                        & number of levels                         &                            & $2$        \\
\bottomrule
  \end{tabular}}
  \caption{Parameters for the overlap solver. Here, {\em dp} denotes double precision and {\em sp} single precision.}
  \label{table:allparms1}
\end{table}

Table~\ref{table:allparms1} summarizes the default settings used for the results reported in Figure~\ref{fig:constant_sign_fct_cost}. The superscripts {\em dp} and {\em sp} indicate that 
we perform the preconditioning in IEEE single precision arithmetic, while the multiplication with $D_N$ within the (F)GMRES 
iteration is done in double precision arithmetic. Such mixed precision approaches are a common further strategy to reduce computing times 
in lattice simulations.

\begin{figure}[htb]
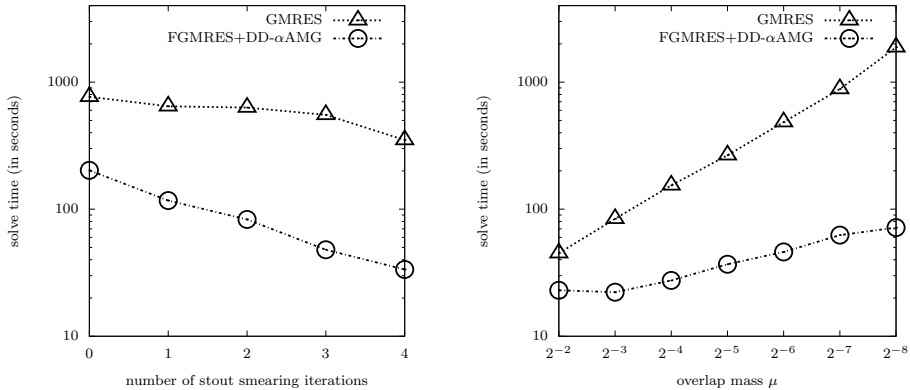

  \begin{minipage}[t]{0.48\textwidth}
  \centering \scalebox{0.68}{\input{./comparison_fixed_sign_cost}}
  \end{minipage} \hfill
  \begin{minipage}[t]{0.48\textwidth}
  \centering \scalebox{0.68}{\input{./fixed_cost_scan_rho_32s3}}
  \end{minipage}
  \caption{Comparison of execution times for preconditioned FGMRES and GMRES. Left: for $0$ to $4$ steps of stout smearing (configuration ID~\ref{JF_32_32}, default value for $\mu$ from Table~\ref{table:allconfs}), right: different overlap masses $\mu$ for configuration ID~\ref{JF_32_32} and 3-step stout smearing.}
  \label{fig:constant_sign_fct_cost}
\end{figure}

For the results reported in Figure~\ref{fig:constant_sign_fct_cost} we tried to keep the cost for a matrix vector multiplication with $D_N$ 
independent of the number of smoothing steps which were applied to the configuration. To do so, we used the $100$th smallest eigenvalue of $\Gamma_5D_W$ for $s=0$ as
a threshold, and deflated all eigenpairs with eigenvalues below this threshold for the configurations with $s>0$. 
The left plot in Figure~\ref{fig:constant_sign_fct_cost} shows that, at fixed default overlap mass $\mu$, we gain a factor of 4 to 10 
in execution time using the preconditioner. The quality of the preconditioning improves with the numbers of smearing steps. 
The right part of Figure~\ref{fig:constant_sign_fct_cost} shows that for smaller values of $\mu$ we can expect an even larger reduction of
the execution time. For the smallest value considered, $\mu = 2^{-8},$ which is realistic for future lattice simulations, the improvement due to preconditioning is a factor of about $25$.


\subsection{Comparison of optimized solvers}



Physics production codes for simulations with the overlap operator use recursive preconditioning as an additional technique to further
reduce the cost for the matrix vector multiplication (MVM) with $D_N$; cf.~\cite{Cu05a}. This means that the FGMRES iteration is preconditioned by using an additional ``inner'' iteration to approximately invert $D_N$, this inner iteration being itself again FGMRES. The point is that we may require only low accuracy for this inner iteration, implying that all MVMs with $\sign(\Gamma_5D_W)$ in the inner iteration may be approximated to low accuracy and  computed in IEEE single precision, only. 

In this framework, we can apply the DD-$\alpha$AMG preconditioner, too, but this time as a preconditioner for the inner FGMRES iteration. In this manner we keep the advantage of needing only a low accuracy approximation to the MVM with $\sign(\Gamma_5D_W)$, while at the same time reducing the 
number of inner iterations and thus the (low accuracy) evaluations of MVMs with $\sign(\Gamma_5D_W)$. 

We denote $\varepsilon_{\inner}$ the residual reduction we ask for in the unpreconditioned inner iteration and $\varepsilon_{\inner}^{\prec}$ the corresponding accuracy required when using the DD-$\alpha$AMG iteration as a preconditioner. The inner iteration converges much faster when we use preconditioning. More accurate solutions in the inner iteration reduce the number of outer iterations and thus the number of costly high precision MVMs with $\sign(\Gamma_5D_W)$. When preconditioning is used for the inner iteration, requiring a higher accuracy in the inner iteration comes at relatively low additional cost. It is therefore advantageous to choose $\varepsilon_{\inner}^{\prec}$ smaller than $\varepsilon_{\inner}$. As an addition to Table~\ref{table:allparms1}, Table~\ref{table:allparms2} lists the default values we used for the inner iteration and which were found to be fairly optimal via numerical testing. 

\begin{table}[htb]
\centering\scalebox{0.9}{\begin{tabular}{llcc}
\toprule
                        & parameter                                            &  notation                  & default     \\
\midrule

inner FGMRES${}^{sp}$            & required reduction of initial residual  & $\varepsilon_{\inner}^{\prec}$    & $10^{-2}$  \\
                        & (with preconditioning)                  &                           &            \\
                        & required reduction of initial residual  & $\varepsilon_{\inner}    $    & $10^{-1}$   \\
                        & (without preconditioning)               &                                   &             \\
                        & relaxation strategy                     & 
                        &  $\frac{\varepsilon_{\inner},\varepsilon_{\inner}^{\prec}}{\|r_k\|}\cdot 10^{-2}$                                       \\
                        & restart length                          & $m_{\restart}^{\inner}$ & $100$       \\
\bottomrule
  \end{tabular}}
  \caption{Parameters for the inner iteration.}
  \label{table:allparms2}
\end{table}

Figure~\ref{fig:recursive} shows results for the solvers optimized in
this way. We consider different sizes for the deflation subspace, i.e., the number of 
smallest eigenvalues which we deflate explicitly. The computation of these eigenvalues (via PARPACK~\cite{wwwPARPACK}) is costly, so that deflating a 
larger number of eigenvalues is efficient only if several system solves with the same overlap operator are to be performed. The figure shows that,
irrespectively from the number of deflated eigenvalues, the preconditioned recursive method outperforms the unpreconditioned method in a 
similar way it did in the non-recursive case considered before. When more smearing steps are applied, the improvement grows; improvement factors 
reach 10 or more. The figure also shows that in the case that we have to solve only one or two linear systems with the same matrix, it is not advisable to use deflation at all, the cost for the computation of the eigenvalues being too large. We attribute this finding at least partly to 
the fact that the thick restart method used to approximate the sign function from \cite{FrGuSc14} is particularly efficient, here.
While all other data in Figure~\ref{fig:recursive} was obtained for configuration ID~\ref{JF_32_32}, 
the rightmost data on the left plot refers to configuration
ID~\ref{BMW_32_32}. We see a similar high efficiency of our
preconditioner as we did for configuration~\ref{JF_32_32} with 3
smearing steps, an observation consistent with the fact that
configuration ID~\ref{BMW_32_32} was also obtained using 3 steps of
(HEX) smearing,  see Table~\ref{table:allconfs}.

\begin{figure}[]
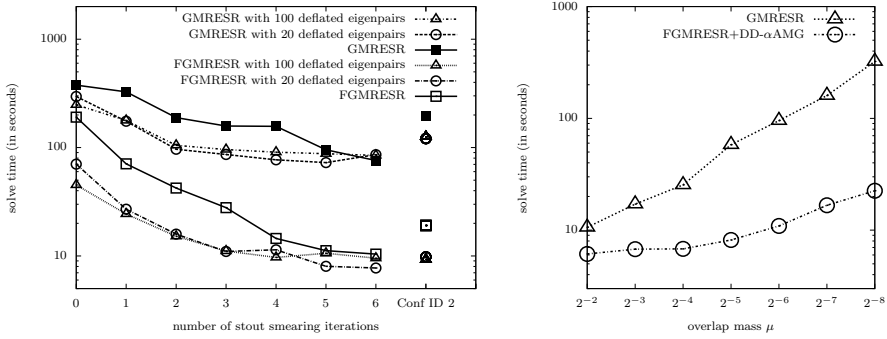

  \centering \scalebox{0.58}{\input{./recursive}}
  \centering \scalebox{0.58}{\input{./recursive_scaling}}
  \caption{Comparison of GMRESR with FGMRESR with different deflation spaces (configuration IDs~\ref{JF_32_32} and \ref{BMW_32_32} with $1,\!024$ processes).}
  \label{fig:recursive}
\end{figure}


\section*{Conclusions} The new, fast adaptive algebraic multigrid solvers for the Wilson-Dirac operator $D_W$ 
now allow to efficiently use this operator as a preconditioner for the overlap operator. We presented a thorough analysis of this auxiliary space preconditioner in the case that $D_W$ is normal. This is not the case in practice, but 
the trend in current simulations in lattice QCD is to reduce the non-normality of $D_W$ as one approaches the continuum limit and smearing techniques are applied. For a state-of-the-art parallel implementation 
and for physically relevant configurations and parameters we showed that the improvements in time to solution gained through the preconditioning are at least a factor of 4 and, typically, more than 10.

\section*{Acknowledgments}

We thank the Budapest-Marseille-Wuppertal collaboration and Jakob Finkenrath for providing configurations. All numerical results were obtained on Juropa at J\"ulich Supercomputing Centre (JSC) through NIC grant HWU12.

\appendix

\section{The entries of $D_W^HD_W-D_WD_W^H$}\label{app:normalitydeviation}
In order to prove Proposition~\ref{Fnorm:prop} we inspect the entries of $D_W^HD_W-D_WD_W^H$.
\newcommand{\pro}{\pi}
We use the notation $\pro_\mu^{\pm}$ for the matrices
\[
  \pro_\mu^{\pm} = \tfrac{1}{2}(I_4 \pm \gamma_\mu), \; \mu=0,\ldots,3.
  \]
  
The relations \eqref{commutativity_rel:eq} between the $\gamma$-matrices show that
each $\pro_\mu^\pm$ is a projection and that, in addition, 
\begin{equation} \label{proj_relations:eq}
   \pro_\mu^{+}\pro_\mu^{-}= \pro_\mu^{-}\pro_\mu^{+} = 0, \; \mu=0,\ldots,3. 
\end{equation}

Considering all 12 variables at each lattice site as an entity, the graph associated with the nearest neighbor coupling 
represented by the matrix $D_W$ is the 4d-torus, and similarly for $D_W^H$. Table~\ref{D_DH:tab} repeats \eqref{DandDH_entries:eq}, 
giving the non-zero entries of a (block) row in $D_W$ and $D_W^H$ in terms of the $12 \times 12$ matrices which couple lattice site $x$ with the sites
$x$ and $x\pm \hat{\mu}$. We use $m$ to denote $m_0+4$ with $m_0$ from \eqref{Wilson-Dirac:eq}  

\begin{table}[]
\begin{center}
\begin{tabular}{l|ll}
                       & \multicolumn{1}{c}{$D$}        & \multicolumn{1}{c}{$D^H$} \\ \hline
$(x,x)$                &  $mI_{12}$                     & $mI_{12}$       \\
$(x,x +\hat{\mu})$   & $-\pro_\mu^{-} \otimes U_\mu(x)$ & $-\pro_\mu^+ \otimes U_\mu(x)$ \\
$(x,x -\hat{\mu})$   & $-\pro_\mu^{+} \otimes U^H_\mu(x-\hat{\mu})$ & $-\pro_\mu^- \otimes U^H_\mu(x-\hat{\mu})$ 
\end{tabular} 
\end{center}
\caption{Coupling terms in $D_W$ and $D_W^H$.\label{D_DH:tab}}
\end{table}

The product $D_W^HD_W$ represents a coupling between nearest and next-to-nearest lattice sites; the coupling $12 \times 12$ matrices
are obtained as the sum over all paths of length two on the torus of the product of the respective coupling matrices in $D_W^H$ and $D_W$.
A similar observation holds for $D_WD_W^H$. 
Table~\ref{D^HD:tab} reports all the entries of $D_W^HD_W$, and we now shortly discuss all the paths of length 2
which contribute to these entries of $D_W^HD_W$.  

For the diagonal position $(x,x)$ we have 21 paths of length 2,  $(x,x) \to (x,x) \to (x,x)$ and $(x,x) \to (x,x\pm \hat{\mu}) \to (x,x), \mu = 0,\ldots,3$. The contribution of each of the latter 20 paths is 0 due to \eqref{proj_relations:eq}.  
For a nearest neighbor $(x,x + \hat{\mu})$ we have the two paths $(x,x) \to (x,x) \to (x,x + \hat{\mu})$ and $(x,x) \to (x,x+\hat{\mu}) \to (x,x+\hat{\mu})$, and similarly in the negative directions. 
For a position $(x,x \pm 2\hat{\mu})$ there is 
only one path $(x,x) \to (x,x \pm \hat{\mu}) \to (x,x \pm 2\hat{\mu})$, with the product of the couplings being 0 due to 
\eqref{proj_relations:eq}. Finally, for the other next-to-nearest neighbors we always have two paths, for example  $(x,x) \to (x,x+\hat{\mu}) \to (x + \hat{\mu} - \hat{\nu})$ and $(x,x) \to (x,x-\hat{\nu}) \to (x + \hat{\mu} - \hat{\nu})$. 
 
The coupling terms in $D_WD_W^H$ are identical to those for $D_W^HD_W$ except that we have to interchange all $\pro_\mu^+$ and $\pro_\mu^-$ as well
as all $\pro_\nu^+$ and $\pro_\nu^-$.

\begin{table} 
\begin{center}
\begin{tabular}{l|l}
$(x,x)$                &  $m^2I_{12}$                    \\ \hline
$(x,x +\hat{\mu})$   & $-m (\pro_\mu^+ + \pro_\mu^-) \otimes U_\mu(x) $ \\
$(x,x -\hat{\mu})$   & $-m (\pro_\mu^+ + \pro_\mu^-) \otimes U_\mu(x-\hat{\mu})$ \\ \hline
$(x,x \pm 2\hat{\mu})$  & $0$                                    \\ \hline
\multicolumn{1}{c|}{$\nu \neq \mu$:}    &    \\
$(x,x+\hat{\mu} + \hat {\nu})$ & $\pro_\mu^-\pro_\nu^+ \otimes U_\mu(x)U_\nu(x+\hat{\mu}) 
                               + \pro_\nu^-\pro_\mu^+ \otimes U_\nu(x)U_\mu(x+\hat{\nu}) $  \\                        
$(x,x+\hat{\mu} - \hat {\nu})$ & $ \pro_\mu^-\pro_\nu^- \otimes U_\mu(x)U^H_\nu(x+\hat{\mu}-\hat{\nu}) 
                               + \pro_\nu^+\pro_\mu^+ \otimes U^H_\nu(x-\hat{\nu})U_\mu(x-\hat{\nu}) $  \\
$(x,x-\hat{\mu} - \hat {\nu})$ & $\pro_\mu^+\pro_\nu^- \otimes U^H_\mu(x-\hat{\mu})U^H_\nu(x-\hat{\mu}-\hat{\nu}) 
   + \pro_\nu^+\pro_\mu^- \otimes U^H_\nu(x-\hat{\nu})U^H_\mu(x-\hat{\nu}-\hat{\mu}) $  
\end{tabular} 
\end{center}
\caption{Coupling terms in $D_W^HD_W$. The coupling terms in $D_WD_W^H$ are obtained by interchanging all $\pro_\mu^+$ and $\pro_\mu^-$ as well
as all $\pro_\nu^+$ and $\pro_\nu^-$.\label{D^HD:tab}}
\end{table}

This shows that in $D_W^HD_W-D_WD_W^H$ the only no-vanishing coupling terms are those at positions 
$(x,x+\hat{\mu} + \hat{\nu})$,  $(x,x+\hat{\mu} - \hat{\nu})$ and $(x,x-\hat{\mu} - \hat{\nu})$ for 
$\mu \neq \nu$. They are given in Table~\ref{non_normality:tab}, where we used the identities 
\begin{equation*}
\begin{array}{rcl}
  \pro_{\mu}^{-}\pro_{\nu}^{-} -
    \pro_{\mu}^{+}\pro_{\nu}^{+} &=& \frac12\left(- \gamma_{\mu} -
      \gamma_{\nu}\right),\\[1ex]
  \pro_{\mu}^{+}\pro_{\nu}^{-} -
    \pro_{\mu}^{-}\pro_{\nu}^{+} &=& \frac12\left(\gamma_{\mu} -
      \gamma_{\nu}\right),\\[1ex]
  \pro_{\mu}^{-}\pro_{\nu}^{+} -
    \pro_{\mu}^{+}\pro_{\nu}^{-} &=& \frac12\left(-\gamma_{\mu} +
      \gamma_{\nu}\right),\\[1ex]
  \pro_{\mu}^{+}\pro_{\nu}^{+} -
    \pro_{\mu}^{-}\pro_{\nu}^{-} &=& \frac12\left(\gamma_{\mu} +
      \gamma_{\nu}\right).\\[1ex]
\end{array}
\end{equation*}
By rearranging the terms we obtain the plaquettes from \eqref{plaquette_def1:eq} and 
\eqref{plaquette_def2:eq}. We exemplify this for position $(x,x+\hat{\mu} + \hat{\nu})$
\begin{eqnarray*} 
  \lefteqn{\hspace*{-3.5ex} \pro_\mu^-\pro_\nu^+ \otimes U_\mu(x)U_\nu(x+\hat{\mu})  + \pro_\nu^-\pro_\mu^+ \otimes U_\nu(x)U_\mu(x+\hat{\nu}) } \\
  \lefteqn{\hspace*{-7ex}\, - \, \left( \pro_\mu^+\pro_\nu^- \otimes U_\mu(x)U_\nu(x+\hat{\mu})  + \pro_\nu^+\pro_\mu^- \otimes U_\nu(x)U_\mu(x+\hat{\nu})  \right)}   \\
  &=&  \tfrac{1}{2}( - \gamma_\mu + \gamma_\nu) \otimes U_\mu(x)U_\nu(x+\hat{\mu})  +  
             \tfrac{1}{2}(\gamma_\mu - \gamma_\nu) \otimes U_\nu(x)U_\mu(x+\hat{\nu})  \\
  &=& \tfrac{1}{2}(-\gamma_\mu + \gamma_\nu) \otimes \left( I_3 - Q_x^{\mu,\nu} \right) U_\mu(x)U_\nu(x+\hat{\mu}) .
\end{eqnarray*}  

\begin{table}
\begin{tabular}{l|l} 
$\mu \neq \nu$: &    \\ \hline 
$(x,x+\hat{\mu} + \hat{\nu})$ & $ \tfrac{1}{2}(-\gamma_\mu + \gamma_\nu) \otimes \left( I_3 - Q_x^{\mu, \nu} \right) U_\mu(x)U_\nu(x+\hat{\mu}) $ \\
$(x,x+\hat{\mu} - \hat{\nu})$ & $\tfrac{1}{2}(-\gamma_\mu - \gamma_\nu) \otimes \left( I_3 - Q_x^{\mu, -\nu} \right) U_\mu(x)U_\nu^H(x+\hat{\mu}-\hat{\nu}) $ \\
$(x,x-\hat{\mu} - \hat{\nu})$ & $\tfrac{1}{2}(\gamma_\mu - \gamma_\nu) \otimes \left( I_3 - Q_x^{-\mu, -\nu} \right) U_\mu^H(x-\hat{\mu})U_\nu^H(x-\hat{\mu}-\hat{\nu}) $ 
\end{tabular}
\caption{Coupling terms in $D_W^HD_W-D_WD_W^H$. \label{non_normality:tab}}
\end{table}

Using the fact that for the Frobenius norm we have
\begin{eqnarray*}
   \| AQ \|_F &=& \|A\|_F \mbox{ whenever $Q$ is unitary (and $AQ$ is defined)}, \\
   \| A \otimes B \|_F  &=& \|A \|_F \cdot \|B \|_F \mbox{ for all $A,B$},
\end{eqnarray*}
we obtain the following for the squares of the Frobenius norms of all the coupling matrices from 
Table~\ref{non_normality:tab}:
\[
\begin{array}{lcl}
 2 \| I - Q_x^{\mu,\nu}\|_F^2 && \mbox{ for position } (x,x+\muh + \nuh), \\
 2 \|I-Q_x^{\mu,-\nu} \|_F^2 &&\mbox{ for position } (x,x+\muh-\nuh),\\
 2 \| I - Q_x^{-\mu,-\nu} \|_F^2 &&\mbox{ for position } (x,x-\muh-\nuh).
\end{array}
\]

Finally for any unitary matrix $Q$ we have 
\[
\| I -Q \|_F^2 = \tr{(I-Q^H)(I-Q)} = 
     2 \cdot \Re(\tr{I-Q}).
\]  

Now we obtain $\|D_W^HD_W-D_WD_W^H\|_F^2$ by summing the squares of the Frobenius norms of all 
coupling matrices. This is a sum over 24$n$ coupling matrices, $n$ being the number of lattice 
sites. As discussed before, groups of four of these coupling matrices refer to the same plaquette 
$Q_x^{\mu,\nu}$  up to conjugation in SU(3), so $\tr{I-Q}$ is the same for these four plaquettes 
$Q$. We can thus ``normalize'' to only consider all possible ``first quadrant'' plaquettes 
$Q_x^{\mu,\nu}$ and obtain  
\[
\| D_W^HD_W - D_WD_W^H \|_F^2 = 4 \sum_x \sum_{\mu < \nu} 2 \cdot 2 \cdot \Re(\tr{I-\pla_x^{\mu,\nu}}).
\]

\bibliographystyle{siam}
\bibliography{overlap_note_cites}

\end{document}